\documentclass{article}
\usepackage{arxiv}
\usepackage[utf8]{inputenc} 
\usepackage[T1]{fontenc}    
\usepackage{hyperref}       
\usepackage{url}            
\usepackage{booktabs}       
\usepackage{amsfonts}       
\usepackage{nicefrac}       
\usepackage{microtype}      
\usepackage{graphicx}
\usepackage{natbib}
\usepackage{doi}
\usepackage{amsmath}
\usepackage{amsthm}
\usepackage{algorithm}
\usepackage{algpseudocode}
\usepackage{float}
\usepackage{enumitem}
\usepackage{graphicx}
\usepackage{epstopdf}
\usepackage{multirow}
\usepackage[flushleft]{threeparttable}

\newtheorem{lem}{Lemma}

\newcommand{\abs}[1]{\left\vert#1\right\vert}

\title{Maximum Cut Algorithms and Upper Bounds for
Planar and Toroidal Graphs}

\author{{Mark Glass}\\
	School of Electrical and Computer Engineering\\
	Tel Aviv University\\
	\texttt{markglass@mail.tau.ac.il} \\
	\And
	{Meir Feder} \\
	School of Electrical and Computer Engineering\\
	Tel Aviv University\\
	\texttt{meir@eng.tau.ac.il} \\
}

\renewcommand{\shorttitle}{Max-Cut Algorithms and Bounds for Planar and Toroidal Graphs}
\RequirePackage{algorithm}[0.1]

\hypersetup{
pdftitle={A template for the arxiv style},
pdfsubject={q-bio.NC, q-bio.QM},
pdfauthor={David S.~Hippocampus, Elias D.~Striatum},
pdfkeywords={First keyword, Second keyword, More},
}

\begin{document}
\maketitle

\begin{abstract}
We demonstrate that the problem of finding the maximum cut of a planar graph with arbitrary weights can be easily mapped to a minimum T-join problem in the absolute dual graph - the dual graph with absolute weights, as opposed to the known mapping to a maximum T-join problem with an empty set in the dual graph. By enabling the use of the shortest paths, this approach allows for the straightforward adaptation of the first efficient Max-Cut algorithm, designed by Hadlock in 1975 for planar graphs with non-negative weights, to handle the general case of planar graphs with arbitrary weights. Furthermore, we prove that applying a planar Max-Cut algorithm to a higher genus graph, such as a toroidal graph, while disregarding its topology, provides an upper bound for its maximum cut. Employing this methodology, we derive upper bounds for the maximum cut across all toroidal graphs within the GSet benchmark. We report that the known maximum cut values for part of those GSet toroidal problems including the three largest instances, which were previously documented in the literature, are the maximum possible because they match their upper bound values. Additionally, we introduce a novel heuristic algorithm for finding Max-Cut of toroidal graphs, which is based on the planar graph algorithm. Applying this algorithm to all seventeen toroidal Max-Cut problems in the GSet benchmark successfully reproduces all the best-known results, and for problem \#62, it yields a new, previously unknown best Max-Cut value.
\end{abstract}


\section{Introduction}
The Maximum Cut (Max-Cut) problem of a graph is a foundational challenge in combinatorial optimization and graph theory.  Given an undirected graph $G = (V, E)$, the objective is to partition the vertex set $V$ into two disjoint subsets, $V_1$ and $V_2$, such that the total weight of edges connecting the two subsets is maximized. Alternatively, a Max-Cut can be specified by the edge set $E_{MC}$ comprising the cut ($E_{MC}\subseteq E$), and $V_{0,1}$ can be derived from it. Note that any partition of the graph nodes into two sets forms a valid cut but not any edge set represents a cut of the graph. While in the general case the problem is NP-hard, its mathematical structure provides a robust framework for modeling systems that seek to maximize "separation" or minimize "frustration" within a network.

In the field of VLSI design, Max-Cut is instrumental in optimizing physical layouts. During the routing phase of chip design, signals must often jump between different metal layers through "vias" (vertical interconnection areas). Minimizing these vias reduces manufacturing costs and improves signal integrity. It is known that via minimization task can be mapped to a Max-Cut formulation \cite{kups55018}.

One of the most profound connections between graph theory and physics lies in the study of Ising problem. In a magnetic system, "spins" (up or down) interact with their neighbors. In a frustrated system (a spin glass), not all interactions can be satisfied simultaneously. Finding the ground state—the configuration with the lowest total energy—is mathematically equivalent to solving a Max-Cut problem on the underlying graph of interactions.

In addition Max-Cut algorithms are very useful in computer vision and image processing (image segmentation and energy minimization), social network analysis, bio-informatics, distributed computing deployment, and logistics optimization.

The paper is organized as follows. Basic definitions and prior art are given in \ref{sec:basic} and \ref{sec:prior}, respectively. In \ref{sec:generalize} we describe the generalization of Hadlock' Max-Cut algorithm for the case of planar graphs with arbitrary weights. In \ref{sec:bound} and \ref{sec:gsetbound} we report on upper bounds for toroidal Max-Cut problems from GSet benchmark. Finally, in   \ref{sec:toroid} we describe a heuristic Max-Cut algorithm for toroidal graph and report its performance across GSet toroidal problems. Conclusions follow in \ref{sec:conclude}.

\section{Basic Definitions}
\label{sec:basic}

This article focuses exclusively on undirected graphs where every edge has a real weight and is part of at least one circuit. For edges that are not part of a circuit, the decision of whether to include them in the Max-Cut edge set is simply determined by the sign of their weight.

We define "positive/negative edges" as edges possessing positive/negative weights, respectively. The "sum of edges" refers to the summation of their weights. The "absolute value of an edge" is defined as the absolute value of its weight. An "odd/even circuit" denotes a circuit containing an odd/even number of edges, regardless of whether their weights are positive or negative. The terms "odd/even circuits" with respect to a specific set of edges signify circuits comprising an odd/even number of edges belonging to that set. An "odd/even positive circuit" is a circuit containing an odd/even number of positive edges. The absolute graph of a graph $G$, denoted by $|G|$, is a graph identical to $G$ except that all edge weights are replaced by their absolute values. The absolute dual graph $|D|$ of graph $G$ is defined as the dual of the absolute graph of $G$.

The symmetric difference of two sets (denoted by $\oplus$) is the union of the elements that belong to only one of the sets:
\begin{align*}
    A \oplus B = (A \cup B) \setminus (A \cap B)
\end{align*}
A number of useful properties of symmetric difference operation are given below:
\begin{itemize}
\item{$B \oplus B  = \emptyset$}
\item{$B \oplus \emptyset = B$}
\item{$(A \oplus B) \oplus C  = A \oplus (B \oplus C)$}
\item{$C = A \oplus B \iff A = C \oplus B$}
\end{itemize}

An edge set in a graph constitutes a cut if and only if every circuit within the graph contains an even number of edges from that set. This is a necessary and sufficient condition for partitioning the vertices into two valid subsets based on the given edge set.
However, this condition can be simplified by applying it only to the circuits in a basic set. Any other circuit in the graph can be expressed as a symmetric difference of a number of basic circuits, and, therefore, its parity is even if the parities of all basic circuits are even. This is because of the property that the symmetric difference of any number of even circuits is always even. For instance, in a planar graph, the face circuits of any planar embedding form a set of basic circuits.

T-join problem of a graph $G$ is defined as finding an edge subset $J$ such that a vertex in $G$ has an odd number of incident edges from $J$ if and only if that vertex belongs to the set $T$. Minimum T-join is a T-join with minimum sum of edges.

\section{Prior Art}
\label{sec:prior}

Max-Cut algorithm design remains an active research area. Recently, several studies have been published on this topic, with new best solutions reported for problems in the GSet benchmark (e.g., see \cite{Chimani_Dahn_2020, Zick23, Zick25, khan2025newbestknownmaxcutsolution}). The GSet is a standardized and widely used benchmark suite for testing algorithm and hardware performance on the Max-Cut problem. Introduced in the early 2000s, it continues to be essential for evaluating classical solvers, quantum annealers, and heuristic algorithms.

1n 1975 Hadlock \cite{Hadlock1975} presented the first polynomial-complexity algorithm for determining the maximum cut in a planar graph, which was originally designed for unweighted graphs and graphs with non-negative edge weights. The algorithm is based on the idea that the maximum cut edge set is the complement of a minimum odd circuits cover. This minimum odd circuits cover corresponds to a minimum odd vertex cover in the dual graph. This problem is then transformed into a minimum weight matching (MWM) problem, which can be solved using Edmonds' Blossom algorithm \cite{Edmonds1965}. This algorithm has a polynomial complexity of $O(n^3)$, where $n$ is the number of odd vertices. The requirement for non-negative weights is necessary because the shortest path search used in the process will encounter infinite loops if negative edges are present in the graph.

In 1990, Barahona \cite{Barahona90} and Mutzel \cite{Mutzel90} independently proposed reducing the Max-Cut problem to an empty set minimum T-join problem in the dual graph. This reduction is applicable even to the general case of a planar graph with arbitrary weights. Subsequent research by \cite{Bieche_1980, Barahona1982,Berman99,LiersPardella2012} introduced graph expansion techniques that allow for the Minimum Weight Matching (MWM) algorithm to be applied directly to the planar extended graph. This yields the minimum T-join solution by contracting the extended graph. This contrasts with the original Hadlock algorithm, where MWM is applied to a complete (non-planar) graph constructed from the shortest path lengths in the dual. The benefit of applying MWM directly to the extended planar graph is a potential theoretical complexity of $O(n^{3/2} \log n)$ \cite{Shih1990,Berman99,LiersPardella2012}. This can be achieved by means of using a divide and conquer approach rooted in the vertex separator theorem \cite{Lipton77, Shih1990}. Practically, the standard Blossom MWM algorithm can also be used on the planar extended graph, resulting in fewer computations compared to using it on a complete graph (see also \cite{Kolmogorov2009} and \cite{CombOpt2018}).

The following section demonstrates that solving the maximum cut problem for a planar graph with arbitrary weights is equivalent to finding a minimum T-join in its absolute dual graph. This equivalence is based on the concept of frustrated circuits with positive and negative edges that was originally described for the Ising problem on a square grid by  Bieche \cite{Bieche_1980} and Barahona \cite{Barahona1982}. We adapt this idea to the case of the maximum cut of a general planar graph.

\section{Generalization of Hadlock Max-Cut Algorithm for Planar Graphs with Arbitrary Weights}
\label{sec:generalize}

In this section we describe a simple generalization of the algorithm defined by Hadlock \cite{Hadlock1975} for the case of a planar graph with arbitrary weights. The concept behind it is to start with the all positive edge set, since the sum of those edges is an upper bound for Max-Cut. Naturally, if there are odd positive circuits, this edge set cannot be a cut of the graph. The idea is to convert this set to a cut of the graph by means of removing/adding edges with positive/negative weights and getting all circuits even w.r.t. the resulting set of edges. The intuition is that in order to get the Max-Cut the choice of which edges will be added/removed is determined by total sum of the absolute values of their weights. Let T be the vertex set in the absolute dual graph corresponding to odd positive circuits in the primal. Then, the minimum T-join in the absolute dual graph can be found in the same way as the odd vertices cover in the original algorithm - via finding shortest paths between each pair of vertices of interest and minimum weight matching of the complete graph. Then, the Max-Cut edge set in the primal graph is the symmetric difference of all positive edge set and edge set corresponding to minimum T-join edge set in the absolute dual. Note, that for the case of a planar graph with all non-negative edges, the algorithm is identical to the original one: the minimum T-join is identical to the minimum odd vertices cover, and the operation of symmetric difference with positive edge set is identical to the complement operation. Let's formulate that rigorously.

Let's denote by a "parity changer" of a circuit set an edge set such that the parity of any circuit edges contained in its symmetric difference with all positive edge set is different from the parity of the positive edges in the circuit if and only if the circuit is in the set.
Let's denote by "minimum parity changer" a parity changer with the minimum sum of absolute values of its edge weights.

Let's denote by "positive evenizer" of a graph an edge set such that its symmetric difference with the all positive edge set contains and even number of edges in any basic circuit of the graph. Essentially, the meaning of positive evenizer is to indicate which positive edges can be omitted from and which negative edges can be added to the set of all positive edges in order to make all basic (and as a consequence all other) circuits even w.r.t. the resulting edge set. Note, that for the case of all non-negative edges, the positive evenizer is an odd circuit cover which is defined in \cite{Hadlock1975}, and vice versa under condition that the complement of the odd circuit cover is a full cut. Since we are looking for minimum odd circuit covers eventually, one which is not a complement of a full cut is not a minimum one and, therefore, is not of interest. Also it is worth noting that positive evenizer is a parity changer of the set of odd positive basic circuits.

Let's denote by "minimum positive evenizer" a positive evenizer with the minimum sum of absolute values of its edge weights.
\begin{lem}\label{lem:cover}
An edge set is a cut of the graph if and only if its symmetric difference with all positive edge set is a positive evenizer.
\end{lem}
\begin{proof}
Let C be a cut. The intersection of any circuit with C is even. Therefore, the edge set which is the symmetric difference of all positive edge set with C is a positive evenizer.

Conversely, if C is a symmetric difference of a positive evenizer with the all positive edge set, an intersection of C with any circuit contains even number of edges by definition and therefore C is a cut.
\end{proof}
Then we can state the following:
\begin{lem}\label{lem:maxcut}
The Max-Cut edge set of a graph is a symmetric differemce of all positive edge set with its minimum positive evenizer.
\end{lem}
\begin{proof}
Following from Lem.\ref{lem:cover}, the value of a cut is equal to sum of all positive edges minus sum of absolute values of positive and negative edges which are in the corresponding positive evenizer. Therefore, minimum sum of absolute values of edges in the positive evenizer results in maximum cut, and vice versa.
\end{proof}

\begin{lem}\label{lem:abs}
An edge set is a minimum positive evenizer in $G$ if and only if the corresponding edge set in $\abs{G}$ is a minimum parity changer of circuits in $\abs{G}$, corresponding to odd positive basic circuits in G.
\end{lem}
\begin{proof}
A minimum positive evenizer in G is equivalent to the minimum parity changer of the set of odd positive basic circuits in G, which is equivalent to the minimum parity changer of the circuit set in $\abs{G}$ corresponding to the odd positive basic circuits in G.
\end{proof} 
\begin{lem}\label{lem:dual}
An edge set is a minimum parity changer of a circuit set in $\abs{G}$ if and only if the corresponding edge set is the minimum T-join of corresponding vertex set T in $\abs{D}$.
\end{lem}
\begin{proof}
The proof is analogous to the one of Theorem 2 in \cite{Hadlock1975}.
\end{proof}

\begin{lem}\label{lem:mc_arb}
An edge set is a Max-Cut of  graph $G$ if and only if its symmetric difference with all positive edge set corresponds to the minimum T-join in the absolute dual graph $\abs{D}$, where T is vertex set corresponding to odd positive circuit set in $G$.
\end{lem}
\begin{proof}
Follows directly from combination of Lemmas \ref{lem:cover},\ref{lem:maxcut},\ref{lem:abs},\ref{lem:dual}.
\end{proof}

The minimum T-join in the absolute dual graph can be found efficiently using minimum weight matching algorithm. The resulting algorithm is as follows. Given the planar graph G with arbitrary weights, identify odd positive face circuits of G, construct $\abs{G}$ which is an absolute graph of $G$ and its dual $\abs{D}$ (if multiple edges connect the same pair of vertices - select the one with minimum absolute value only), and find a minimum T-join in $\abs{D}$ of the vertex set corresponding to odd positive basic circuits in $G$. The latter stage can be done in the same way as in the original algorithm by \cite{Hadlock1975}, since the absolute dual $\abs{D}$ contains non-negative edges only. The Max-Cut edge set, then, is the symmetric difference of all positive edge set of the graph and the edge set corresponding to minimum T-join in the absolute dual. The top level of the planar Mux-Cut algorithm is given in Alg.\ref{algo1}. 
\begin{algorithm}[H]
\caption{Finding Max-Cut of a Planar Graph with Arbitrary Weights}
\label{algo1}
\begin{algorithmic}[1] 
\Require $G=\{E,V\}$ \Comment{graph edges and vertices}
\Ensure $E_{MC} \subseteq E$ \Comment {Max-Cut}
\State Find odd positive basic circuits $C$ of $G$. \Comment{face circuits with odd \# of pos. edges}
\State Get absolute graph $\abs{G}$. \Comment{Same as $G$ but with absolute values of the weights}
\State Get $\abs{D}$ (if multiple edges connect the same pair of vertices - select the one with minimum absolute value only). \Comment{the dual graph of $\abs{G}$}
\State Identify vertices $T$ in $\abs{D}$ corresponding to $C$ in $G$.
\State Build $M$ - a list containing shortest paths in $\abs{D}$ between all possible pairs of vertices in $T$.
\State Build matrix $P$, containing
lengths of all paths in $M$.
\State Build matrix $Q$ according to:
$Q=max(P)+1-P$.
\State Calculate maximum weight matching of the complete graph $Q$ using Blossom algorithm, which gives the minimum weight matching (MWM) of $P$.
\State Build set $S$ which is the union of all edges in $\abs{D}$ from paths in $M$ corresponding to MWM.
\State Build set $R$ of edges in $G$, corresponding to edges in $S$.
\State Max-Cut edge set is the symmetric difference of $R$ and all positive edge set.
\end{algorithmic}
\end{algorithm}

The comparison between the original Hadlock's algorithm and the new generalized algorithm is given on Fig.\ref{fig:maxcut_algo}. One can see that the difference between the two algorithms is minor. Also, in the case of a planar graph with all non-negative weights, the generalized algorithm is equivalent to the original one.
\begin{figure}[H]
    \centering
    \includegraphics[width=0.9\linewidth]{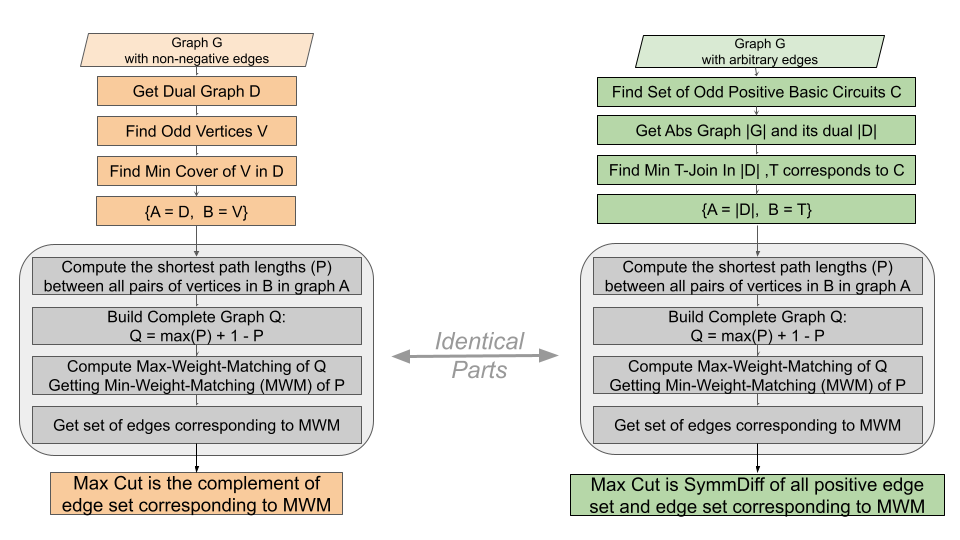}
    \caption{Workflow comparison of finding Max-Cut: the original algorithm by Hadlock \cite{Hadlock1975} for planar graphs with non-negative weights (left) and the new generalized algorithm for planar graphs with arbitrary weights (right).}
    \label{fig:maxcut_algo}
\end{figure}
An example of the algorithm specified above applied to a simple graph is given on Fig.\ref{fig:Example1}. On the resulting diagram (VI) the Max-Cut edge set is indicated by  thick lines.
\begin{figure}[H]
    \centering
    \includegraphics[width=1.0\linewidth]{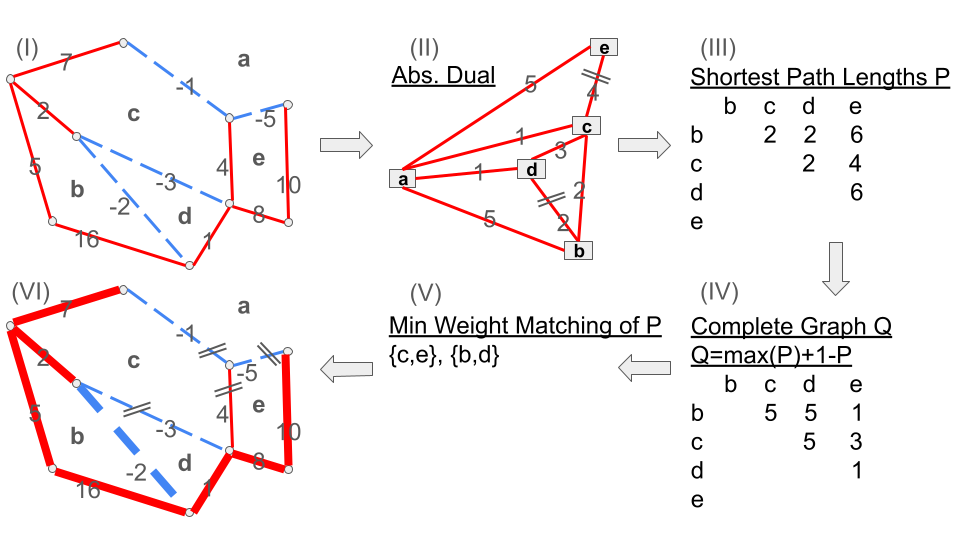}
    \caption{Example of Max-Cut algorithm applied to a planar graph with arbitrary weights. (I) The given graph, the odd positive circuits are all face circuits: b,c,d,e. (II) Absolute dual graph. (III) Shortest path lengths complete graph P. (IV) Complete graph Q with shortest path lengths subtracted from their maximum plus 1. (V) Minimum weight matching of vertices in complete graph P. (VI) Max-Cut edges indicated by thick lines.} 
    \label{fig:Example1}
\end{figure}

\section{Upper Bound for Max-Cut of Toroidal Graphs}
\label{sec:bound}
In a toroidal graph the face circuits and the two additional circuits - vertical and horizontal basic circuits (VBC and HBC), as depicted on Fig.\ref{fig:toro}, comprise the set of basic circuits of the graph (see \cite{Bieche_1980, Barahona1982}). Any other circuit can be expressed as a symmetric difference of a number of circuits from this basic set. According to Lem.\ref{lem:maxcut}, the Max-Cut edge set of a toroidal graph $G$ is a symmetric difference of all positive edge set and the minimum positive evenizer. Also the minimum positive evenizer is equivalent to the minimum parity changer of the odd positive basic circuits of $G$.

\begin{lem}\label{mcbound}
Applying a planar graph maximum cut algorithm to a toroidal graph, without accounting for its higher genus topology, provides an upper bound for its maximum cut.
\end{lem}
\begin{proof}
In a toroidal graph, the Max-Cut edge set is the symmetric difference of the positive edge set and a minimum positive evenizer, let's denote it by S. Following the definition of $S$, which is also a minimum parity changer of odd positive basic circuits of the graph, it is also correct to say that this edge set S is in the group of the edge sets which are parity changers of odd positive face circuits only (disregarding VBC and HBC). The minimum sum of edge weights absolute values for an edge set from this group is found by the Max-Cut algorithm for planar graphs. In other words, the minimum parity changer of odd positive face circuits is the symmetric difference of all positive edge set and edge set given by the planar Max-Cut algorithm applied to the toroidal graph. Therefore the value of Max-Cut of a toroidal graph cannot be higher than the value given by the algorithm for finding Max-Cut of a planar graph, applied to the toroidal graph disregarding its topology.
\end{proof}

This claim is true for graphs of higher genera as well. We conclude that for a graph of any genus we can calculate an upper bound of its Max-Cut by means of applying a planar graph Max-Cut algorithm.
\begin{figure}[H]
    \centering    \includegraphics[width=1.0\linewidth]{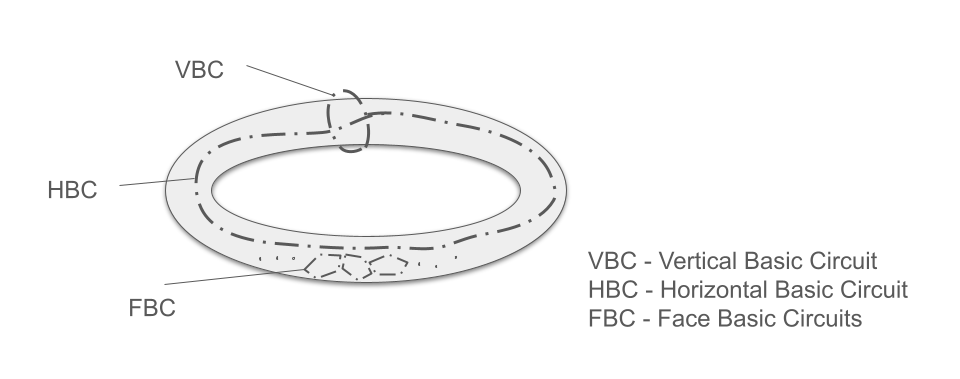}
    \caption{Basic circuits of a toroidal graph: VBC, HBC and FBC} 
    \label{fig:toro}
\end{figure}
\section{Max-Cut Bound for Toroidal GSet Problems}\label{sec:gsetbound}
The upper bounds calculated by applying planar Max-Cut algorithm to toroidal GSet benchmark problems, are given in Tbl.\ref{mcbounds}.
Here the run times refer to calculation of the main part of planar Max-Cut algorithm, which is the Blossom algorithm for finding MWM. The MWM was tested on the complete (Full) input graph with all vertices connected to each other, and on a graph containing only edges with weights not greater than 4 (Partial), in both cases the resulting bound was the same for each toroidal GSet problem (the length of the longest shortest paths used by MWM in the full case is 4). The run times were measured on a home PC equipped with Intel I5 2.6GHz processor and 32GB RAM, except the full calculation of the last largest problem which was done on a stronger station. Note, that the measured run times of the full graph are approximately proportional to the 3rd power of number of odd nodes in the problems, which is consistent with its polynomial complexity order. The known Max-Cut results are taken from \cite{GSetGithub, Zick23, Zick25}.

\section{Max-Cut of a Toroidal Graph}
\label{sec:toroid}
This section details a heuristic algorithm giving an approximated solution to Max-Cut problems in toroidal graphs with arbitrary weights. The approach utilizes the planar Max-Cut algorithm, incorporating an additional step to address the evenness condition of a vertical basic circuit (VBC), while the evenness condition of a horizontal basic circuit (HBC) is satisfied implicitly during the algorithm's final stage.\\
We begin with applying the planar Max-Cut algorithm (Alg.\ref{algo1}) to the toroidal graph. Then we select a random VBC and satisfy its evenness condition w.r.t. the edge set corresponding to MWM output using the following procedure:
\begin{enumerate}
    \item Find the Minimum Weight Matching (MWM) of vertices in the absolute dual of the toroidal graph corresponding to its odd positive face circuits, same as in Max-Cut algorithm (this also can be formulated as minimum T-join problem).
    \item Check the parity of shortest path edges from the MWM crossing the VBC.
    \item If the parity is the same as the parity of the VBC positive edges, do one of the following two changes to MWM input shortest paths list:
    \begin{itemize}
        \item No changes are made or
        \item Remove two longest of the shortest paths crossing the VBC, fix the rest.
    \end{itemize}
    \item If the parities of the crossing edges and VBC postive edges are different, do one of the following two changes to MWM input shortest paths list:
    \begin{itemize}
        \item Remove the longest of the shortest paths crossing the VBC or
        \item Remove the second longest of the shortest paths crossing the VBC,
    \end{itemize}
    and fix the rest of shortest paths crossing the VBC.
\end{enumerate}
The Max-Cut algorithm then proceeds starting from MWM until completion.\\
The procedure described above yields an edge set that satisfies the evenness conditions for all face circuits and the VBC.
However, this edge set does not necessarily represent a valid cut of the toroid because the evenness condition of the HBC was not considered.
We then color the nodes into two colors based on this edge set, stopping when all nodes are colored and not necessarily all edges from the set are used. Since any partition of all graph nodes into two groups constitutes a valid cut of the graph, we use the total weight of edges connecting the two colored groups as a candidate for Max-Cut value.
This toroidal Max-Cut algorithm is summarized in Alg.\ref{ToroidMC}.
\begin{table}
\caption{Known results, upper bounds and our results for Max-Cut in toroidal GSet problems, run times of the minimum weight matching (MWM) - the main part of the planar Max-Cut algorithm in terms of computational complexity, run times of our toroidal Max-Cut algorithm. The tests were done on a home PC equipped with Intel I5 2.6GHz processor and 32GB RAM.}
\label{mcbounds}
\small
\begin{center}
\begin{tabular}{|c|c|c|c|c|c|c|c|c|c|c|}
\hline
\multirow{ 2}{*}

   GSet & \#Nodes /        & \#Odd        & Known       &  Upper        &  \multicolumn{2}{c|}{MWM time [sec]}   &Our                   &Time      \\\cline{6-7}
    \#  & Diameter         & Circuits     & Result      &  Bound        &  Full     & Partial                    &Result                &[sec]     \\\hline
11      & 800/8            & 434          & \bf564      &  \bf564       &     8     &  1                         &\bf564                &2         \\\hline
12      & 800/16           & 394          & 556         &  \bf558       &     6     &  1                         &556                   &44        \\\hline
13      & 800/32           & 384          & 582         &  \bf583       &     4     &  1                         &582                   &31        \\\hline
32      & 2E3/20           & 1014         & 1410        &  \bf1414      &     86    &  2                         &1410                  &156       \\\hline
33      & 2E3/25           & 1004         & 1382        &  \bf1384      &     98    &  2                         &1382                  &181       \\\hline
34      & 2E3/40           & 975          & 1384        &  \bf1386      &     79    &  2                         &1384                  &208       \\\hline
48      & 3E3/60           & 0            & \bf6000     &  \bf6000      &     0(*)  &  0(*)                      &\bf6000               &7         \\\hline
49      & 3E3/100          & 0            & \bf6000     &  \bf6000      &     0(*)  &  0(*)                      &\bf6000               &7         \\\hline
50      & 3E3/120          & 0            & 5880        &  \bf6000      &     0(*)  &  0(*)                      &5880                  &7         \\\hline
57      & 5E3/50           & 2520         & \bf3494     &  \bf3494      &     1609  & 11                         &\bf3494               &596       \\\hline
62      & 7E3/70           & 3488         & 4870        &  \bf4872      &     4960  & 24                         &\bf\underline{4872}   &72        \\\hline
65      & 8E3/80           & 4002         & \bf5562     &  \bf5562      &     7156  & 33                         &\bf5562               &753       \\\hline
66      & 9E3/100          & 4456         & 6364        &  \bf 6366     &     9853  & 36                         &6364                  &7650      \\\hline
67      & 10E3/100         & 5020         & 6950        &  \bf6952      &     15495 & 50                         &6950                  &9632      \\\hline
72      & 10E3/100         & 5000         & \bf7008     &  \bf7008      &     15057 & 51                         &\bf7008               &147       \\\hline
77      & 14E3/100         & 7016         & \bf9940     &  \bf9940      &     44507 & 103                        &\bf9940               &287       \\\hline
81      & 20E3/100         & 9978         & \bf14060    &  \bf14060     &    N/A\tnote{1}(**) & 212              &\bf14060              &14651     \\\hline
\end{tabular}
\begin{tablenotes}
    \item[1] (*) There are no odd positive face circuits in problems \#48,49,50, therefore MWM algorithm is not invoked.
    \item[2] (**) The last problem of the largest size was tested on a different computer - a stronger one, due to increased usage of RAM, therefore the runtime is irrelevant for comparison.
  \end{tablenotes}\end{center}
\end{table}

\begin{algorithm}
\caption{Finding Max-Cut of a Toroidal Graph with Arbitrary Weights}
\label{ToroidMC}
\begin{algorithmic}[1] 
\Require $G=\{E,V\}$ \Comment{Graph edges and vertices}
\Ensure $E_{MC} \subseteq E$ \Comment {Max-Cut}

\State Run planar Max-Cut algorithm until getting MWM.  \Comment{Alg.\ref{algo1}}
\State Find a vertical basic circuit (VBC) of the toroid.
\If{The parities of the shortest paths  edges in MWM crossing the VBC and VBC positive edges are the same}
       \State Either do nothing or remove two longest shortest crossing paths from MWM input and fix the remaining crossing paths.
\Else \State{Remove one of two longest shortest paths crossing the VBC and fix the rest.}
\EndIf
\State Continue planar Max-Cut algorithm until completion, getting the edge set $E_0$.
\State Color the graph's nodes using two colors, guided by the edge sets $E_0$ and $\overline{E_0}$. The coloring process is finished when every node has been colored, regardless of whether all edges have been processed. Specifically, for any given edge, its connected nodes must share the same color if and only if that edge is in $\overline{E_0}$.
\State Calculate the sum of all edges connecting nodes from different color groups and take it as a candidate for Max-Cut value.
\State Rerun the algorithm with a different VBC until a number of candidates for Max-Cut is found, select the maximum one as the answer.
\end{algorithmic}
\end{algorithm}

The toroidal Max-Cut algorithm, as specified here, was applied to all seventeen toroidal problems in the GSet benchmark. For each problem, we utilized twenty four VBCs with random starting indices. The algorithm was terminated upon finding a cut value equal to the upper bound. The results were consistently at least as good as the best known results reported in the literature, and for problem \#62, a new, previously unknown best value was discovered. Run times varied from a few seconds to several hours. \ref{mcbounds} summarizes the details of this test.

\section{Conclusions}
\label{sec:conclude}
We conducted a thorough review of Maximum Cut (Max-Cut) algorithms for planar graphs. Hadlock's original Max-Cut algorithm \cite{Hadlock1975} was previously thought to only apply to planar graphs with non-negative edge weights. We proved that finding the Max-Cut of a planar graph with arbitrary weights can be reduced to a problem of finding a minimum T-join in a planar graph with non-negative weights (T-join definition can be found \ref{sec:basic}). This key finding enabled us to design a simple adaptation of the original algorithm by Hadlock \cite{Hadlock1975} to the more general case of planar graphs with arbitrary weights. Furthermore, we showed that a Max-Cut algorithm designed for planar graphs can be applied to non-planar graphs to determine upper bounds for their Max-Cut values. This technique can be particularly valuable for evaluating the performance of different Max-Cut solvers on toroidal and higher genus graphs. We used this methodology to calculate upper bounds for GSet toroidal problems and discovered that the best known Max-Cut values for several of them, including the three largest ones, are best possible because they matched their upper bounds. Finally, we introduced a heuristic algorithm for finding approximated Max-Cut in toroidal graphs, which is derived from the planar graph algorithm, and documented its superior performance on the toroidal problems within the GSet benchmark. This algorithm found solutions with all the  best known Max-Cut values and a new best  Max-Cut solution for problem \#62, which had not been previously known (Tbl.\ref{mcbounds}).
We identify the following directions as potentially promising for future research:
\begin{itemize}
    \item {Minimum T-join Problem}: developing a faster algorithm for the minimum T-join problem in graphs with non-negative weights. Specifically, implementing a method to cancel a portion of the shortest paths (e.g., those longer than a certain threshold) in the complete graph used for Minimum Weight Matching (MWM). This simplification could significantly reduce MWM runtime and benefit all approaches to finding the maximum cut in planar graphs.
    \item {Maximum Cut Algorithms and Vertex Separation}: advancing maximum cut algorithms by utilizing the vertex separation theorem. Further research could reduce the polynomial order of complexity for maximum cut solutions, allowing for the computation of maximum cuts for larger problem instances.
    \item {Max-Cut Upper Bound for Non-Planar Graphs}: the Max-Cut upper bound calculated using the planar algorithm described here can serve as a bound criterion for branch-and-bound type solvers applied to the Max-Cut problem in non-planar graphs.
\end{itemize}

\section*{Acknowledgments}
We thank Khen Cohen and Prof. Yaron Oz for fruitful discussions on the topics of Ising problem and graph theory.

\section*{Code Availability}
The Python code scripts and GSet problem data used in this study, including files for Max-Cut solution, verification, bounds calculation, and results, are available at: 
\href{https://github.com/mrkgs/MaxCutPlanarAndToroidal} {https://github.com/mrkgs/MaxCutPlanarAndToroidal}

\bibliographystyle{unsrt}

\end{document}